 \documentclass[letterpaper, 10 pt, conference]{ieeeconf}  

\IEEEoverridecommandlockouts                              

\let\labelindent\relax
\usepackage{enumitem}
\usepackage{cite}
\usepackage{amsmath, amssymb}
\usepackage{xcolor}
\usepackage{algorithm}
\usepackage{algpseudocode}
\usepackage{amsthm}    
\usepackage{amsmath}
\usepackage{amssymb} 

\newtheorem{lemma}{Lemma}

\usepackage{tikz}
\usetikzlibrary{positioning, shapes}
\usepackage{mathtools} 
\usepackage{todonotes}
\usepackage{amsthm}  
\usepackage{amsmath}
\usepackage{amssymb}
\usepackage{xcolor}
\usepackage{graphicx} 
\usepackage{subcaption}
\usepackage{subcaption}
\usepackage[font=small]{caption} 
\def\BibTeX{{\rm B\kern-.05em{\sc i\kern-.025em b}\kern-.08em
    T\kern-.1667em\lower.7ex\hbox{E}\kern-.125emX}}
\usepackage[framemethod=TikZ]{mdframed}
\usepackage{mdframed}
\usepackage{algorithm}
\usepackage{algpseudocode}
\theoremstyle{remark}
\newtheorem{remark}{Remark}
\makeatletter
\renewcommand{\ALG@beginalgorithmic}{\setcounter{ALG@line}{0}}
\makeatother

\renewcommand{\baselinestretch}{0.95}



\begin{document}

\title{Multi-Agent Event-Triggered LQG Control under Shared Communication Constraints}
\author{Zahra Hashemi and Dipankar Maity \vspace{-5mm}
\thanks{
This research is supported by the National Science Foundation CAREER Award 2443349.}
\thanks{
The authors are with the Department of Electrical and Computer Engineering at the University of North Carolina at Charlotte, NC, USA, 28223. (e-mails: {\tt \{zahrahashemi1, dmaity\}@charlotte.edu}).}
}

\maketitle
\thispagestyle{empty}
\pagestyle{empty}

\begin{abstract}
This letter studies event-triggered linear-quadratic-Gaussian (LQG) control for multi-agent systems sharing a communication network with limited per-step capacity. Although the agent dynamics are decoupled, the communication decisions are coupled through the shared network constraint, leading to a constrained multi-agent scheduling problem. We show that the optimal control law remains certainty-equivalent and decouples across agents through independent finite-horizon Riccati recursions, whereas the transmission schedule remains globally coupled. Based on this structure, we develop a centralized receding-horizon scheduling framework and reformulate the resulting problem as a mixed-integer linear program (MILP) using a closed-form characterization of the estimation-error covariance. To improve scalability, we derive a window-based skip-pruning condition that safely fixes consecutive transmission decisions to zero before solving the MILP, and we propose an auction-inspired scheduler based on one-step transmission-benefit scores. Numerical results show that the proposed model predictive control (MPC) scheduler achieves the best trade-off between control performance and communication cost, while the auction-based scheduler attains performance close to MPC with substantially lower computational complexity.

\end{abstract}

\section{Introduction}
\label{sec:introduction}

Networked control systems often operate under stringent communication
constraints, especially when multiple agents share a wireless medium.
In such settings, continuous transmission of all sensor measurements may
be infeasible due to bandwidth, energy, or access limitations. This has
motivated extensive work on \emph{event-triggered} and
\emph{communication-aware} control, where information is transmitted only
when it is sufficiently valuable for closed-loop
performance~\cite{maity2019optimal,wang2010event,dimarogonas2011distributed,hashemi2026event}.
In LQG settings, communication-constrained control has been studied from
several complementary perspectives, including finite-rate or costly
communication, minimal-feedback architectures, and explicit
communication-performance trade-offs
~\cite{molin2009lqg,molin2012optimality,borkar1997lqg,
maity2020minimal,sabag2023reducing}. Related studies have also considered
distributed and multi-agent LQG control under communication and
information constraints
~\cite{yilmaz2018distributed,perez2025distributed,
lee2024linear,kashyap2024optimal}.

A fundamental challenge arises when the communication medium is shared
by multiple agents~\cite{perez2025distributed,ge2019distributed,
mamduhi2025network}. Even if the agent dynamics are independent, the
communication decisions become coupled through the network capacity:
at each time step, only a limited number of transmissions can be
accommodated. As illustrated in Fig.~\ref{fig:network_capacity}, this
creates a multi-agent scheduling problem in which the network must decide
which agents should be granted access at each instant. The resulting
design problem combines stochastic control, event-triggered communication,
and combinatorial resource allocation over time.

\begin{figure}[t]
    \centering
    \includegraphics[width=0.8\columnwidth,trim={4.4cm 3.8cm 4.4cm 3. 4cm},clip]{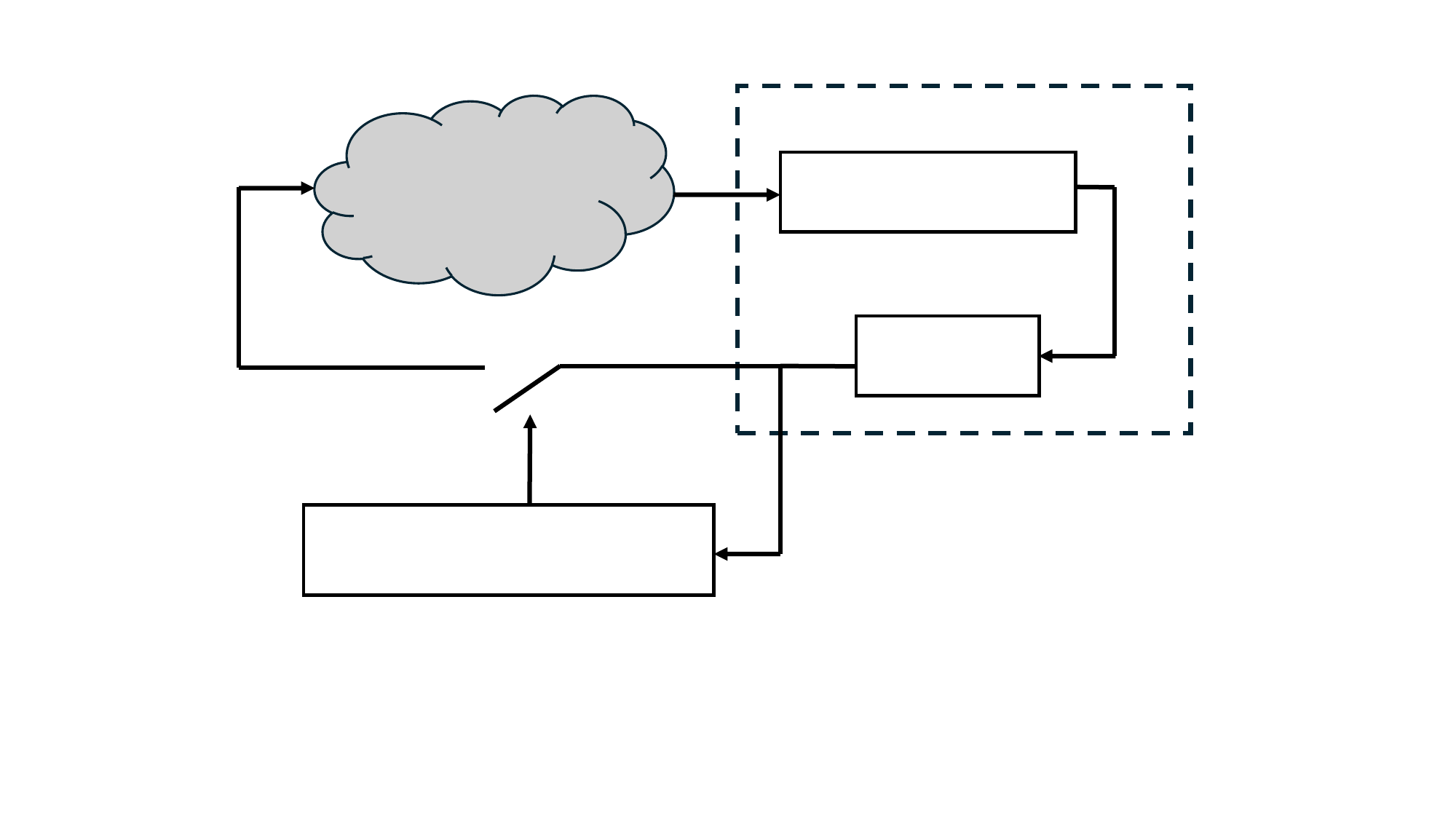}
    \begin{picture}(0,0)
            \put (-178,16) {Central Scheduler}
            \put (-198,91) {$z_{i,k}$}
            \put (-35,93) {$u_{i,k}$}
            \put (-98,60) {$x_{i,k}$}
            \put (-135,40) {$\theta_{i,k}$}
        \put (-93,83) { Controller  $i$}
        \put (-82,53) { Agent  $i$}
        \put (-165,83) {\textbf{Network}}
                \put (-98,96) {Sub-System $i$}
                
    \end{picture}          
    \caption{Networked control architecture for subsystem $i$ under a centralized network scheduler.}  \vspace{-4mm} \label{fig:network_capacity}
\end{figure}
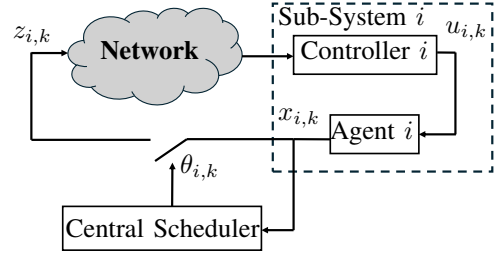

For single-agent LQG systems with costly communication, Molin
\emph{et al.}~\cite{molin2009lqg,molin2012optimality} established
that the optimal controller is certainty-equivalent, reducing the joint
control and communication design problem to scheduling over the
estimation-error process. Building on this structure, our recent
work~\cite{hashemi2025linear} considered a single stochastic system
and derived an equivalent MILP
reformulation of the communication scheduling problem, together with
one-step optimality certificates for transmission and skipping decisions.
The present work extends this setting to multiple heterogeneous systems
sharing a communication network with limited per-step capacity.
Although the agent dynamics and control laws remain decoupled, the
transmission decisions are coupled through the shared capacity constraint,
since allocating a communication slot to one agent may prevent another
agent from transmitting at the same time. Consequently, the per-agent
scheduling problems can no longer be solved independently, and the
scheduler must jointly allocate the available communication resources
across agents.

To address this shared-capacity allocation problem, this paper develops a
centralized receding-horizon scheduling framework for multi-agent
event-triggered LQG control. Building on the single-agent covariance
representation and MILP reformulation in
\cite{hashemi2025linear}, the formulation is extended to multiple
heterogeneous agents whose transmission decisions are jointly coupled
through the per-step network-capacity constraint. The centralized MILP
is complemented by a window-based pruning mechanism and an
auction-inspired low-complexity allocation rule.

The main contributions of this paper are as follows.
First, we formulate the shared-capacity multi-agent scheduling problem
as a centralized MILP, where otherwise independent agent-level
transmission decisions become coupled through the common network
constraint.
Second, we derive a \emph{window-based skip-pruning condition} that
provides a sufficient condition for fixing multiple consecutive
transmission decisions to zero without loss of optimality, thereby
reducing the online search space.
Third, we develop an \emph{auction-inspired scheduling rule} based on
per-agent transmission-benefit scores, providing a computationally
lightweight alternative to the centralized MILP.

Finally, numerical results compare the centralized MPC, one-shot,
auction-based, and periodic schedulers, illustrating the trade-off
between control performance, communication cost, and computational
complexity. In particular, the MPC scheduler achieves the best overall
performance, while the auction-based heuristic attains comparable
performance with substantially lower online computation.

\textit{Notation:} Let $[N]\triangleq\{1,\ldots,N\}$. The sets of real numbers and nonnegative integers are denoted by $\mathbb{R}$ and $\mathbb{N}_0$, respectively. The notation $A\preceq B$ means $B-A\succeq 0$. For any $Q\succeq 0$, define $\|x\|_Q^2:=x^\top Qx$. We use $\operatorname{tr}(\cdot)$ for the trace and $\mathbb{E}[\cdot]$ for expectation.

\section{Problem Description} \label{sec:problem}
Let us consider a discrete-time stochastic \emph{multi-agent} networked control system comprised of $N$ agents. 
For each agent $i\in[N]$, the local state $x_{i,k} \in \mathbb{R}^{n_i}$ evolves according~to
\begin{equation}
\small
x_{i,k+1} \;=\; A_i x_{i,k} + B_i u_{i,k}\;+\; w_{i,k},
\label{eq:ma_dynamics}
\end{equation}
\begin{equation}
\small
z_{i,k} \;=\;
\begin{cases}
x_{i,k}, & \theta_{i,k} = 1,\\
\varnothing, & \theta_{i,k} = 0,
\end{cases}
\label{eq:ma_measurement}
\end{equation}
where $A_i \in \mathbb{R}^{n_i \times n_i}$ and $B_i \in \mathbb{R}^{n_i \times m_i}$ are constant local system matrices. The control input of agent $i$ is denoted by $u_{i,k} \in \mathbb{R}^{m_i}$, and $w_{i,k} \in \mathbb{R}^{n_i}$ is a zero-mean stochastic disturbance. In \eqref{eq:ma_measurement}, the variable $z_{i,k}$ denotes the information received by the corresponding controller at time $k$, which depends on the communication decision variable $\theta_{i,k} \in \{0,1\}$. Specifically, when $\theta_{i,k}=1$, the controller receives the local state $x_{i,k}$; otherwise, no packet is received, i.e., $z_{i,k}=\varnothing$. 

This formulation enables event-triggered communication across the network, where each agent selectively transmits its state based on a triggering rule, thereby reducing communication while retaining closed-loop performance.

The initial conditions are modeled as random variables with finite first and second moments. 
In particular, for each agent {\small $i\in[N]$}, the initial state $x_{i,0}$ has finite mean $\mu_{i,0}$ and covariance {\small $\Sigma_{i,0}$}. 
The process noises {\small $\{w_{i,k}\}_{k\ge 0}$} are assumed to be independent and identically distributed (i.i.d.) across time, with $w_{i,k}$ being zero-mean with finite covariance {\small $\Sigma_i^{w}$}. 
Moreover, the initial states {\small $\{x_{i,0}\}_{i=1}^N$} are statistically independent of the noise sequences {\small $\{w_{i,k}\}_{k\ge 0}$} for all $i$, and the disturbances are independent across agents and across time, i.e., {\small $w_{i,k}$} is independent of $w_{i',k'}$ for {\small $(i,k)\neq(i',k')$}.

To formalize the controller's access to measurement data over time, we introduce the \emph{information set} available at each time step. 
We assume the decision-maker also remembers its past control actions. 
For each agent $i\in[N]$, define the local information set at time $k\ge 1$ as
\[
\mathcal{Z}_{i,k} \;\triangleq\; \{\, z_{i,0:k},\ \theta_{i,0:k},\ u_{i,0:k-1}\,\},
\]
with initial information set $\mathcal{Z}_{i,0} \triangleq \{z_{i,0},\theta_{i,0}\}$.
Accordingly, the information set evolves as
\begin{equation}
\small
\mathcal{Z}_{i,k+1} \;=\; \mathcal{Z}_{i,k}\cup\{\, z_{i,k+1},\ \theta_{i,k+1},\ u_{i,k}\,\},
\quad \forall i\in[N],\ \forall k\ge 0.
\label{eq:info_set_multiagent}
\end{equation}
We consider the system evolution over a finite horizon of $T$ steps. The primary objective is to jointly design a collection of control policies
$\mu \triangleq \{\mu_i\}_{i=1}^N$ and an event-triggered communication scheduling policy
$\Theta \triangleq \{\Theta_i\}_{i=1}^N$, where $\Theta_i \triangleq \theta_{i,0:T-1}$, to minimize the expected finite-horizon cost.
For each agent $i\in[N]$, we define the individual performance index
\begin{equation}
\small
\begin{multlined}
J_i(\mu_i,\Theta_i)
=
\mathbb{E}\!\Bigg[
\sum_{k=0}^{T-1}
\Big(\|x_{i,k}\|^2_{Q_i}+\|u_{i,k}\|^2_{R_i}+\lambda_i\theta_{i,k}\Big)\\
+\|x_{i,T}\|^2_{Q_{i, T}}
\Bigg],
\end{multlined}
\label{eq:Ji_cost}
\end{equation}
where the expectation $\mathbb{E}[\cdot]$ is taken with respect to the randomness in the initial condition and process disturbance for agent $i$.
Here, $Q_i \succeq 0$ and $Q_{i,T} \succeq 0$ weight the state of agent~$i$; $R_i \succ 0$ weights its control effort; and $\lambda_i>0$ penalizes communication by agent $i$ through $\theta_{i,k}\in\{0,1\}$.

The network-level objective is defined as the sum of individual agent costs,
\begin{equation}
\small
J(\mu,\Theta) \;\triangleq\; \sum_{i=1}^{N} J_i(\mu_i,\Theta_i),
\label{eq:J_sum}
\end{equation}
so that the total number of transmissions over the horizon,
{\small $\sum_{i=1}^{N}\sum_{k=0}^{T-1}\theta_{i,k}$}, is directly penalized with agent-dependent weights $\lambda_i$.
Therefore, \eqref{eq:J_sum} captures the trade-off between closed-loop performance and communication effort over $[0,T]$.
To model limited network resources, we impose a per-time-step capacity constraint: at each time $k$, the network can serve at most $C$ agents simultaneously. Formally, the joint design problem is
\begin{align}
\min_{\mu,\Theta}\quad 
& J(\mu,\Theta)  \label{eq:ma_joint_opt} \\
\text{s.t.}\quad 
& \sum_{i=1}^{N} \theta_{i,k} \le C, 
&& \forall k=0,\dots,T-1. \label{eq:ma_joint_opt_cap}
\end{align} 

\section{Optimal Controller and Scheduler}
\label{sec:controller_scheduler}

\subsection{Certainty-equivalent optimal control (per agent)}
\label{subsec:control}

The certainty-equivalence structure used in this subsection follows from
the single-system results of
\cite{molin2009lqg,molin2012optimality} and is applied independently
to each agent. Since the agent dynamics and per-agent control costs are
decoupled, the shared network-capacity constraint affects the
communication allocation but does not alter the individual Riccati
recursions. Accordingly, for each agent $i$, the optimal controller is
\begin{equation}
u_{i,k}=-L_{i,k}\,
\mathbb{E}[x_{i,k}\mid\mathcal{Z}_{i,k}],
\qquad k=0,\ldots,T-1,
\label{eq:molin-control}
\end{equation}
where $\mathcal{Z}_{i,k}$ denotes the information available at time $k$. The gain $L_{i,k}$ is obtained from the standard finite-horizon Riccati recursion:
\begin{subequations}
\small
\begin{align}
L_{i,k} &= S_{i,k}^{-1} B_i^\top P_{i,k+1} A_i,\\
S_{i,k} &= R_i + B_i^\top P_{i,k+1} B_i,\\
P_{i,k} &= A_i^\top P_{i,k+1}A_i + Q_i
- A_i^\top\! P_{i,k+1} B_i S_{i,k}^{-1} B_i^\top\! P_{i,k+1} A_i,
\end{align}
\label{eq:riccati_multi}
\end{subequations}
with terminal condition {\small$P_{i,T}=Q_{i,T}$} and {\small$P_{i,k}\succeq 0$}. Let
\[\small
\hat{x}_{i,k} \;\triangleq\; \mathbb{E}[x_{i,k}\mid \mathcal{Z}_{i,k}]
\]
denote the conditional state estimate. Under event-triggered measurements, the estimate either resets to the true state when a transmission occurs or evolves according to the open-loop predictor otherwise:
\[\small
\hat{x}_{i,k} =
\begin{cases}
x_{i,k}, & \theta_{i,k}=1,\\
A_i\hat{x}_{i,k-1}+B_i u_{i,k-1}, & \theta_{i,k}=0.
\end{cases}
\]

Substituting~\eqref{eq:molin-control} into the quadratic objective yields an equivalent criterion in which the communication schedule $\Theta_i$ is the only remaining decision variable~\cite{molin2012optimality}:
\begin{equation}
\small
J_i(\mu_i^*,\Theta_i)
=
J_{i,\mathrm{const}}
+
\mathbb{E}\!\bigg[\sum_{k=0}^{T-1}\|e_{i,k}\|^2_{\Gamma_{i,k}}\bigg]+\lambda_i\,\mathbb{E}\!\bigg[\sum_{k=0}^{T-1}\!\!\theta_{i,k}\bigg].
\label{eq:cost-function2}
\end{equation}
where
\[
e_{i,k} \;\triangleq\; x_{i,k}-\hat{x}_{i,k}
\]
is the estimation error, and
\[\small
\Gamma_{i,k} \;\triangleq\; L_{i,k}^\top S_{i,k} L_{i,k}
\]
is the weighting matrix associated with the estimation-error penalty at time $k$. The constant term
\[\small
J_{i,\mathrm{const}}
\;\triangleq\;
\mathbb{E}\!\big[x_{i,0}^\top P_{i,0} x_{i,0}\big]
+
\sum_{k=0}^{T-1}\mathbb{E}\!\big[w_{i,k}^\top P_{i,k+1} w_{i,k}\big]
\]is independent of the schedule $\Theta_i$~\cite{molin2012optimality}. Moreover, the estimation error evolves according to~\cite[Theorem~1]{molin2009lqg}
\begin{equation}
\small
e_{i,k+1} = (1-\theta_{i,k+1})\big(A_i e_{i,k}+w_{i,k}\big),
\label{eq:error-dynamics}
\end{equation}
with $e_{i,0}=(1-\theta_{i,0})(x_{i,0} - \mathbb{E}[x_{i,0}])$. Hence, a transmission resets the next-step estimation error to zero, whereas in the absence of communication the error propagates open-loop and accumulates process noise.

Summing~\eqref{eq:cost-function2} over $i$ shows that the remaining design variable is the \emph{joint transmission schedule} $\{\theta_{i,k}\}_{k\ge0}$, which is coupled across agents through the shared network-capacity constraint~\eqref{eq:ma_joint_opt_cap}.
\subsection{Centralized MPC scheduling and MILP reformulation}

For a single system, the scheduler-side prediction-error representation,
closed-form covariance expansion, and corresponding MILP reformulation
were developed in \cite{hashemi2025linear}. We adopt this per-agent
representation here and extend it to the shared-capacity multi-agent
setting. The key distinction is that, although each agent's error
evolution can be represented independently, the transmission decisions
must be determined jointly because they are coupled through the common
per-step network-capacity constraint. 
Let
\[\small
\bar\theta_{i,k} \;\triangleq\; 1-\theta_{i,k},
\]
so that $\bar\theta_{i,k}=0$ indicates that a packet is received (hence $e_{i,k}=0$), whereas $\bar\theta_{i,k}=1$ indicates that no update is received and prediction is used instead.
Define, for each agent, the scheduler-side one-step prediction error (innovation) as
\begin{align*}
\small
e^s_{i,k}
&:= x_{i,k}-A_i\hat{x}_{i,k-1}-B_i u_{i,k-1}
 \;=\; A_i e_{i,k-1}+w_{i,k-1},\\
e^s_{i,0}
&:= x_{i,0}-\mathbb{E}[x_{i,0}].
\end{align*}
Then, for all $k$,
\[\small
e_{i,k}=\bar\theta_{i,k}\,e^s_{i,k}.
\]
Given the realized $\{e^s_{i,k}\}_{i=1}^N$, the scheduler solves, at time $k$, an optimization problem over the remaining horizon:
{\small\begin{align}
\min_{\{\bar\theta_{i,t|k}\}} \quad
& \sum_{i=1}^N \mathbb{E}\!\left[
\sum_{t=k}^{T-1}\|e_{i,t}\|^2_{\Gamma_{i,t}}
+\lambda_i\bigl(1-\bar\theta_{i,t|k}\bigr)
\right] \nonumber\\
\text{s.t.}\quad
& e_{i,t+1}=\bar\theta_{i,t+1|k}\bigl(A_i e_{i,t}+w_{i,t}\bigr),
\quad t=k,\ldots,T-2, \nonumber\\
& e_{i,k}=\bar\theta_{i,k|k}\,e^s_{i,k},
\quad \forall i\in [N], \nonumber\\
& \sum_{i=1}^N \bigl(1-\bar\theta_{i,t|k}\bigr)\le C,
\quad t=k,\ldots,T-1, \label{eq:cap_mpc}\\
& \bar\theta_{i,t|k}\in\{0,1\},
\quad \forall i\in [N],\ t=k,\ldots,T-1.
\label{eq:mpc_multi}
\end{align}}

The scheduler applies only the first decision $\bar\theta_{i,k|k}^*$ (equivalently, $\theta_{i,k|k}^*$) and repeats the optimization at time $k+1$. For a fixed sequence
\[
\Theta_{i,k}:=\{\theta_{i,t|k}\}_{t=k}^{T-1},
\]
and $t\ge k$, define
\[
\Sigma_{i,t}:=\mathbb{E}\!\left[e_{i,t}e_{i,t}^\top \mid \Theta_{i,k},e^s_{i,k}\right],
\qquad
\Sigma_{i,k}^s \triangleq e^s_{i,k}(e^s_{i,k})^\top.
\]
For each $\tau\in\{k,\ldots,t\}$, define
\begin{align}
\small
\mu_{i,t,\tau}
&:= \prod_{s=\tau}^{t}\bar\theta_{i,s|k},
\label{eq:mu_def_multi}\\[3pt]
G_{i,t,\tau}
&:=
\begin{cases}
A_i^{t-k}\Sigma_{i,k}^s\bigl(A_i^{t-k}\bigr)^\top, & \tau=k,\\[2pt]
A_i^{t-\tau}\Sigma_i^{w}\bigl(A_i^{t-\tau}\bigr)^\top, & \tau\ge k+1,
\end{cases}
\label{eq:G_def_multi}\\[3pt]
g_{i,t,\tau}
&:= \operatorname{tr}\!\bigl(\Gamma_{i,t}G_{i,t,\tau}\bigr).
\label{eq:g_def_multi}
\end{align}
Then, one may verify from~\cite{hashemi2025linear}  that
\begin{align}
\Sigma_{i,t}
&= \sum_{\tau=k}^{t}\mu_{i,t,\tau}G_{i,t,\tau}.
\label{eq:Sigma_sum_multi}
\end{align}
and therefore, 
\[ \small
\mathbb{E}\!\left[e_{i,t}^\top\Gamma_{i,t}e_{i,t}\mid \Theta_{i,k},e^s_{i,k}\right]
=
\operatorname{tr}(\Gamma_{i,t}\Sigma_{i,t})
=
\sum_{\tau=k}^{t}\mu_{i,t,\tau} g_{i,t,\tau}.
\]
To obtain a mixed-integer linear formulation, the previously defined binary variables $\mu_{i,t,\tau}$ are enforced through standard linearization constraints. Since $\mu_{i,t,\tau}$ represents the logical AND over the interval $\{\tau,\ldots,t\}$, the MPC scheduling problem at time $k$ can be written as the following centralized~MILP:
\vspace{-6mm}

{\small\begin{align}
\begin{split}\label{eq:milp_multi}
\min_{\{\bar\theta_{i,t|k}\},\,\{\mu_{i,t,\tau}\}} 
& \sum_{i=1}^N\sum_{t=k}^{T-1}\sum_{\tau=k}^{t} \mu_{i,t,\tau}\, g_{i,t,\tau}  + \sum_{i=1}^N\!\! \lambda_i\!\sum_{t=k}^{T-1}\bigl(1-\bar\theta_{i,t|k}\bigr) \\
\text{s.t.}\quad
& \bar\theta_{i,t|k}\in\{0,1\},\ \mu_{i,t,\tau}\in\{0,1\},\\
& \mu_{i,t,\tau}\le \bar\theta_{i,s|k},
\quad \forall s\in\{\tau,\ldots,t\},\ \forall i,\\
& \mu_{i,t,\tau}\ge \sum_{s=\tau}^{t}\bar\theta_{i,s|k}-(t-\tau),
\quad \forall i,\\
& \sum_{i=1}^N \bigl(1-\bar\theta_{i,t|k}\bigr)\le C,
\quad t=k,\ldots,T-1.
\end{split}
\end{align}}
\vspace{-3mm}

\begin{remark}[Implementation structure and scalability]
\label{rem:implementation_scalability}
The proposed design separates naturally into two computational layers:
(i) \emph{offline} per-agent Riccati recursions~\eqref{eq:riccati_multi}, which compute the gains $L_{i,k},\Gamma_{i,k}$ independently for each agent, and 
(ii) \emph{online} centralized MILP solves~\eqref{eq:milp_multi} at each time step $k$, which determine a network-feasible transmission schedule under the capacity constraint~\eqref{eq:cap_mpc}.
\end{remark}

The centralized MILP contains binary scheduling variables for every
agent and every remaining time step, together with auxiliary binary
variables used to linearize products of scheduling decisions. To reduce
this online combinatorial burden, we next derive a window-based
skip-pruning condition. The result provides a sufficient certificate
under which several consecutive transmission decisions of an agent can
be fixed to ``skip'' before solving the MILP. Hence, the corresponding
binary decisions need not be optimized online, and the associated
product variables and linearization constraints can be simplified,
reducing the effective size of the scheduling problem.
\begin{lemma}
\label{lem:window_safe_skip_pruning}
    Let us define ellipsoid {\small $\mathcal{E}_{i,k,r} = \{x\in \mathbb{R}^{n_i} \mid \|x\|^2_{S_{i,k}} < \tilde\lambda_{i, k,r} \}$}, where $S_k$ is the shape matrix of the ellipsoid and $\tilde\lambda_{i,k,r}$ is the level set:  
\begin{equation*}
\small
    S_{i,k}=\sum_{t=k}^{T-1}\left(A_i^{t-k}\right)^\top \Gamma_t A_i^{t-k},
\end{equation*}
\begin{equation*}
\small
\tilde \lambda_{i,k,r}
=
\lambda_i
-
\sum_{\tau=k+1}^{k+r}\sum_{t=\tau}^{T-1}
\operatorname{tr}\!\left(
\Gamma_{i,t} A_i^{\,t-\tau}\Sigma_i^{w}(A_i^{\,t-\tau})^\top
\right).
\end{equation*}
\vspace{-1mm}
If $e^s_{i,k} \in \mathcal{E}_r$, then $\theta_{i,k+\ell|k} = 0$ optimal for all $\ell = 0, 1,\ldots, r$.
\end{lemma}

\begin{proof}
     Let $\theta_{k:T-1}$ be an arbitrary schedule for the horizon $[k, T-1]$. 
     An $r$-window \textit{skip-transformed} version of this schedule is defined as {\small $\theta^{\text{skip}, r}_{k: T-1} = [\underbrace{0,\ldots,0}_{k: k+r}, \theta_{k+r+1:T-1}]$}, where the first $r+1$ transmissions are set to skip and the rest are the same as the original $\theta_{k:T-1}$. 
     Suppose $\theta_{k:T-1} \ne \theta^{\text{skip}, r}_{k: T-1}$, i.e., at least one transmission happens during the window $[k, k+r]$ in the schedule $\theta_{k:k+r}$. For \(\tau\in\{k,\dots,k+r\}\), we decompose
    \vspace{-2mm}
\begin{equation*}
 \small
     \mu_{i,t,\tau}
 = \prod_{s=\tau}^{t}\bar\theta_{i,s|k} = \underbrace{\prod_{s=\tau}^{k+r}\bar\theta_{i,s|k}}_{\triangleq \mu_{i,t,\tau}^{[k,k+r]}} \underbrace{\prod_{s=k+r+1}^{t}\bar\theta_{i,s|k}
 }_{\triangleq \mu_{i,t,\tau}^{[k+r+1,t]}}\,,
 \end{equation*}
 where the first factor depends on the decisions in the window and the second on the decisions after \(k+r\). 
 For any $\tau > k+r$, we have $\mu_{i,t,\tau}^{[k,k+r]} = 1$ and $\mu_{i,t,\tau}^{[k+r+1,t]} = \mu_{i,t,\tau}$
    
    Schedule $\theta^{\text{skip}, r}_{k: T-1}$ yields for all $t \le k+r$ that $\mu_{i,t,\tau}^{[k,k+r]} = 1$ and consequently,  $\mu_{i,t,\tau} = \mu_{i,t,\tau}^{[k+r+1,t]}$. 
    Therefore, one may~verify \vspace{-4mm}
    
    {\small
     \begin{align*}
        J_i&(\theta^{\text{skip}, r}_{k: T-1})  = \sum_{\tau=k}^{k+r}\sum_{t=\tau}^{T-1} \mu_{i,t,\tau}^{[k+r+1,t]}g_{i,t,\tau} \\& \qquad  \qquad + \sum_{\tau=k+r+1}^{T-1}\sum_{t=\tau}^{T-1}\mu_{i,t,\tau}\,g_{i,t,\tau} + \lambda_i \sum_{t=k+r+1}^{T-1} \theta_{t|k} \\
     \end{align*}
     \begin{align*}
          & = \underbrace{\sum_{\tau=k}^{k+r}\sum_{t=\tau}^{T-1}\mu_{i,t,\tau}\,g_{i,t,\tau} +\!\!\! \sum_{\tau=k+r+1}^{T-1}\sum_{t=\tau}^{T-1}\mu_{i,t,\tau}\,g_{i,t,\tau} + \lambda_i \sum_{t=k}^{T-1} \theta_{t|k}}_{J_i(\theta_{k:T-1})} \\
         &\quad + \sum_{\tau=k}^{k+r}\sum_{t=\tau}^{T-1} (1-\mu_{i,t,\tau}^{[k,k+r]})\,\mu_{i,t,\tau}^{[k+r+1,t]} g_{i,t,\tau} - \lambda_i \sum_{t=k}^{k+r} \theta_{t|k} \\
         & \le J_i(\theta_{k:T-1}) + \sum_{\tau=k}^{k+r}\sum_{t=k}^{T-1}  g_{i,t,\tau} - \lambda_i \sum_{t=k}^{k+r} \theta_{t|k} \\
        & \overset{(\dagger)}{\le} J_i(\theta_{k:T-1}) + \sum_{\tau=k}^{k+r}\sum_{t=k}^{T-1}  g_{i,t,\tau} - \lambda_i \\
        & \overset{(\ddagger)}{<} J_i(\theta_{k:T-1})
    \end{align*}}
    
 where $(\dagger)$ follows due to $\theta_{k:T-1} \ne \theta^{\text{skip}, r}_{k: T-1}$ and therefore, at least one transmission occurs in $[k, k+r]$ for $\theta_{k:T-1}$; and $(\ddagger)$ follows from the lemma statement that $e^s_{i,k} \in \mathcal{E}_r$. 
 Consequently, skipped transmission is optimal for the window $[k, k+r]$ regardless of the rest of transmission schedule.
 \end{proof}

\begin{remark}[Window-based pruning]
\label{rem:certificate_pruning}
Lemma~\ref{lem:window_safe_skip_pruning} extends the
one-step skip certificate in~\cite[Theorem~1]{hashemi2025linear} to a
multi-step window. In particular, $r=0$ recovers a sufficient condition
for fixing the current decision $\theta_{i,k|k}=0$. For $r>0$, the result
certifies that the entire block
$ 
\theta_{i,k|k},\theta_{i,k+1|k},\ldots,\theta_{i,k+r|k}
$
can be fixed to zero without loss of optimality.

Thus, the lemma can be used as a preprocessing step before
solving the MILP: whenever the certificate holds, the corresponding
scheduling variables are removed from the online search, and the
associated auxiliary product variables and linearization constraints can
be simplified. The benefit becomes more significant when the certificate
holds for longer windows or for several agents simultaneously, since a
larger portion of the binary search space is eliminated before the MILP
solver is invoked.
\end{remark}
\subsection{Auction-inspired Capacity Allocation}\label{subsec:game}

To reduce the online computational burden of the centralized MILP in~\eqref{eq:milp_multi}, we introduce an auction-inspired scheduling mechanism based on a one-step approximation of the centralized objective. The main idea is to isolate the part of the MILP cost that depends on the current transmission decisions and use it to define a per-agent transmission-benefit score. The network then allocates the limited communication slots to the agents with the largest estimated benefits.

\paragraph{Agent-side transmission-benefit score}
We focus on the portion of the MILP objective that depends on the current transmission decision at time \(k\). The contribution of agent \(i\) to the MILP objective can be decomposed as
\[\small
J_i = J_{i,k}^{\mathrm{dep}} + J_{i,k}^{\mathrm{ind}},
\]
where
\[\small
J_{i,k}^{\mathrm{dep}}
=
\sum_{t=k}^{T-1}\mu_{i,t,k}\,g_{i,t,k}
+\lambda_i(1-\bar\theta_{i,k|k}),
\]
\[\small
J_{i,k}^{\mathrm{ind}}
=
\sum_{\tau=k+1}^{T-1}\sum_{t=\tau}^{T-1}\mu_{i,t,\tau}\,g_{i,t,\tau}
+
\lambda_i\sum_{t=k+1}^{T-1}(1-\bar\theta_{i,t|k}).
\]
Here, \(J_{i,k}^{\mathrm{dep}}\) collects all terms that depend on \(\bar\theta_{i,k|k}\), whereas \(J_{i,k}^{\mathrm{ind}}\) collects the remaining terms. Since
\[
\mu_{i,k,k}=\bar\theta_{i,k|k},
\qquad
\mu_{i,t,k}=\bar\theta_{i,k|k}\mu_{i,t,k+1},
\]
it follows that 
{\small
\[
J_i
=
\bar\theta_{i,k|k}
\bigg(
g_{i,k,k}+\sum_{t=k+1}^{T-1}\mu_{i,t,k+1}\,g_{i,t,k}-\lambda_i
\bigg)
+\lambda_i+J_{i,k}^{\mathrm{ind}}.
\]
}Using \(0\le \mu_{i,t,k+1}\le 1\), we obtain the upper bound
\begin{equation}
\small
 J_i
\le
\bar\theta_{i,k|k}
\bigg(
\sum_{t=k}^{T-1} g_{i,t,k}-\lambda_i
\bigg)
+\lambda_i+ J_{i,k}^{\mathrm{ind}}.  
\label{upper}
\end{equation}
Since our objective is to minimize the cost, we use the first-transmission-dependent upper bound in~\eqref{upper} as a surrogate objective and define the transmission-benefit score as
\begin{equation}
b_{i,k}\triangleq \sum_{t=k}^{T-1} g_{i,t,k}-\lambda_i.
\label{eq:benefit_score}
\end{equation}

\paragraph{Network objective (minimum cost form).}
Using the scores \(\{b_{i,k}\}_{i=1}^N\), a one-step network-level approximation of the centralized scheduling problem can be written as
\begin{align}
\min_{\bar\theta_{i,k|k}\in\{0,1\}} \quad
& \sum_{i=1}^N \bar\theta_{i,k|k}\, b_{i,k}
\label{eq:auction_efficiency_min}\\
\text{s.t.}\quad
& \sum_{i=1}^N (1-\bar\theta_{i,k|k})\le C.
\nonumber
\end{align}
Since each transmission consumes one identical unit of network capacity,
\eqref{eq:auction_efficiency_min} is solved by selecting the agents with
the $C$ largest positive transmission-benefit scores, or all positively
scoring agents if fewer than $C$ such agents exist.
Agents with $b_{i,k}\leq 0$ are not scheduled, since
transmitting such agents does not decrease the surrogate objective.
The corresponding decentralized benefit computation and centralized
slot assignment are summarized in Algorithm~\ref{alg:auction}.

\begin{algorithm}[t]
\caption{Auction-Inspired Scheduling at Time \(k\)}
\label{alg:auction}
\small
\begin{algorithmic}[1]
\For{each agent \(i=1,\dots,N\) (in parallel)}
    \State Compute \(g_{i,t,k}\) for \(t=k,\dots,T-1\)
    \State Compute transmission-benefit score 
    $
    b_{i,k}
    $ in \eqref{eq:benefit_score}
\EndFor
\State Scheduler selects \(\{\theta_{i,k|k}\}_{i=1}^N\) by choosing the top-\(C\) positive scores.
\State Apply $\theta_{i,k|k}$, set $k\gets k+1$, and repeat.
\end{algorithmic}
\end{algorithm}

The centralized MPC-MILP scheduler optimizes a multi-step cost-to-go and
thus serves as a natural performance benchmark. Its complexity, however,
grows rapidly with the number of agents and the horizon length because of
the increasing number of binary variables. The auction-based scheduler
offers a scalable alternative: each agent computes a scalar
transmission-benefit score, and the network allocates the $C$ available
slots to the agents with the largest scores. This reduces the online
decision process from a horizon-dependent mixed-integer optimization
\eqref{eq:milp_multi} to a simple sorting operation, yielding
substantially lower computational burden. A window-based transmission-benefit score could be used to capture
longer-term effects, at the expense of increased online computational
complexity.

A further advantage of the auction-based scheduler is its low information requirement. Unlike the centralized MPC-MILP \eqref{eq:milp_multi}, the benefit-driven version requires each agent to communicate only a scalar transmission-benefit score at time \(k\). The network can then allocate capacity without access to full trajectories or internal optimization variables. Consequently, the resulting rule is computationally lightweight, communication-efficient, and naturally suited to distributed implementation. Although this one-step approximation is generally suboptimal and we do yet have a theoretical suboptimality bound, the simulations show that its performance remains close to that of the centralized MPC-MILP scheduler.
A tighter approximation could be obtained by incorporating more than one future step using Lemma~\ref{lem:window_safe_skip_pruning}, at the cost of increased online complexity compared to \eqref{eq:auction_efficiency_min}.
\begin{remark}[Heterogeneous packet sizes]
\label{rem:heterogeneous_packets}
If transmissions have heterogeneous packet sizes, let $\ell_i>0$ denote
the resource units required by agent $i$. Then the capacity constraint becomes
\[
\sum_{i=1}^{N} \ell_i \theta_{i,k} \le C,
\qquad k=0,\dots,T-1.
\]
The centralized MILP and the one-step selection problem
\eqref{eq:auction_efficiency_min} remain unchanged except for this weighted
capacity constraint. In this case, the ratio $b_{i,k}/\ell_i$ can be used
to prioritize agents according to their benefit per unit of communication
resource. Hence, heterogeneous packet sizes can be accommodated without
altering the overall framework.
\end{remark}

\section{Simulation Results}
\label{sec:simulation}

We consider $N=10$ independent discrete-time linear stochastic
subsystems with heterogeneous matrices $A_i$ and~$B_i$. Nine subsystems
are open-loop stable while the tenth subsystem is open-loop unstable. This setting highlights the role of communication
scheduling when agents have different sensitivities to estimation errors.
Each subsystem is controlled by a finite-horizon LQG controller with
\[
Q_i=\mathrm{diag}(8,2),\qquad
R_i=2,\qquad
Q_{T,i}=\mathrm{diag}(10,3).
\]
The shared network can accommodate at most $C=5$ transmissions per time
step.  Each transmission incurs a communication cost \(\lambda_i\), with heterogeneous penalties in \([5,25]\).
The total performance index is
\[
J=J_{\mathrm{control}}+J_{\mathrm{comm}},
\qquad
J_{\mathrm{comm}}
=
\sum_{k=0}^{T-1}\sum_{i=1}^{N}
\lambda_i\theta_{i,k}.
\]

We compare four scheduling strategies: a \textit{receding-horizon
scheduler} that solves the mixed-integer problem
\eqref{eq:milp_multi} at each time step and applies only the first
decision; a \textit{one-shot scheduler} that solves the same problem
once at the beginning of the horizon; an \textit{auction-based
heuristic} solving \eqref{eq:auction_efficiency_min}, in which each
agent computes a local score $b_{i,k}$ and the network serves at most
$C$ agents with the largest positive scores; and a \textit{periodic
round-robin scheduler} satisfying the capacity constraint, where agents
are partitioned into groups of size at most $C$ and served cyclically.
For $N=10$ and $C=5$, the periodic scheduler consists of two groups of
five agents and therefore has a period of two time steps.

We consider a horizon of $T=20$. The initial states and process
disturbances are zero-mean Gaussian random variables with heterogeneous
covariance matrices. Results are averaged over $100$ Monte Carlo (MC) runs.
Within each MC trial, all four scheduling methods are evaluated
using the same initial-state and disturbance realizations to enable a
paired comparison. The average transmission decisions $\theta_{i,k}$
are also recorded to illustrate communication patterns across agents
and time.

Fig.~\ref{fig:cost_comparison} compares the average costs of the four
scheduling strategies, with the total cost decomposed into control and
communication components.
\begin{figure}[t]
\centering
\includegraphics[width=0.60\linewidth,trim=150 290 175 300,clip]
{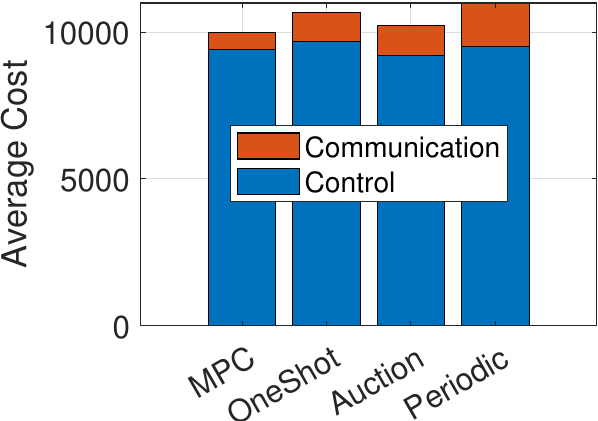}
\caption{Average control and comm. costs over $100$ MC trials.}
\label{fig:cost_comparison} \vspace{-4 mm}
\end{figure}

As shown in Fig.~\ref{fig:cost_comparison}, the centralized MPC
scheduler achieves the lowest average total cost, $J=10002$. The
one-shot, auction-based, and periodic schedulers achieve average total
costs of $10664$, $10241$, and $10993$, respectively. Relative to MPC,
these correspond to increases of approximately $6.62\%$, $2.39\%$, and
$9.90\%$, respectively. The auction-based scheduler therefore attains
performance particularly close to that of the centralized MPC.

The cost decomposition further illustrates the trade-off between control
performance and communication usage. The auction-based scheduler achieves
a slightly lower average control cost than MPC ($9205.1$ versus $9398$),
but incurs a higher communication cost ($1035.7$ versus $604.1$).
Accordingly, MPC attains the lower total objective by communicating more
selectively. On average, MPC uses $2.108$ transmissions per step, compared
with $3.454$ for one-shot, $3.492$ for auction, and $5$ for periodic.
Thus, the periodic policy fully utilizes the network capacity, whereas
the adaptive methods transmit only when communication is sufficiently valuable.

The one-shot scheduler performs worse than MPC and auction because its
schedule is fixed at the start of the horizon and cannot adapt to later
innovations and disturbances. In contrast, MPC and auction update their
communication decisions online. Overall, adaptive scheduling provides a
better balance between control performance and communication cost than
the fixed one-shot and periodic policies.

\begin{figure}[t]
\centering

\begin{subfigure}{0.22\textwidth}
    \centering
    \includegraphics[width=0.95\linewidth,page=1,
    trim={130 280 150 297},clip]
    {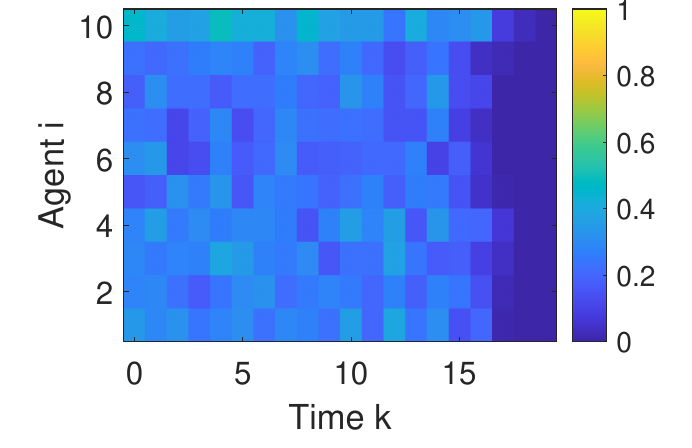}
    \caption{MPC schedule}
\end{subfigure}
\hfill
\begin{subfigure}{0.22\textwidth}
    \centering
    \includegraphics[width=0.95\linewidth,page=1,
    trim={130 280 150 297},clip]
    {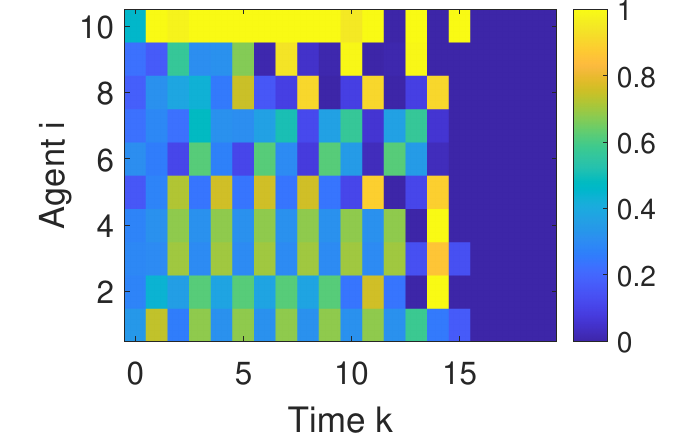}
    \caption{One-shot schedule}
\end{subfigure}

\vspace{0.1cm}

\begin{subfigure}{0.22\textwidth}
    \centering
    \includegraphics[width=0.95\linewidth,page=1,
    trim={130 280 150 297},clip]
    {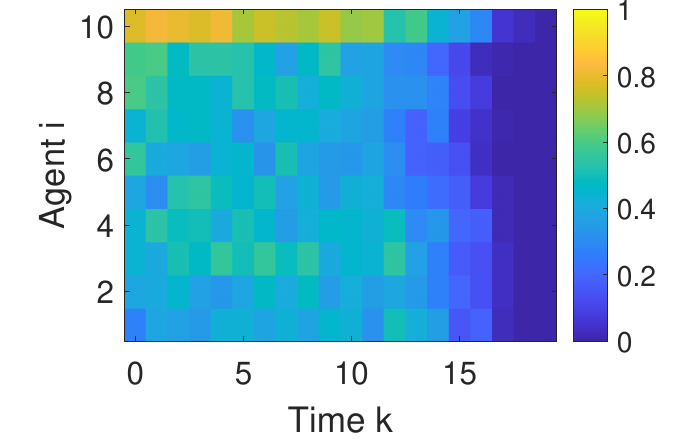}
    \caption{Auction schedule}
\end{subfigure}
\hfill
\begin{subfigure}{0.22\textwidth}
    \centering
    \includegraphics[width=0.95\linewidth,page=1,
    trim={130 280 150 297},clip]
    {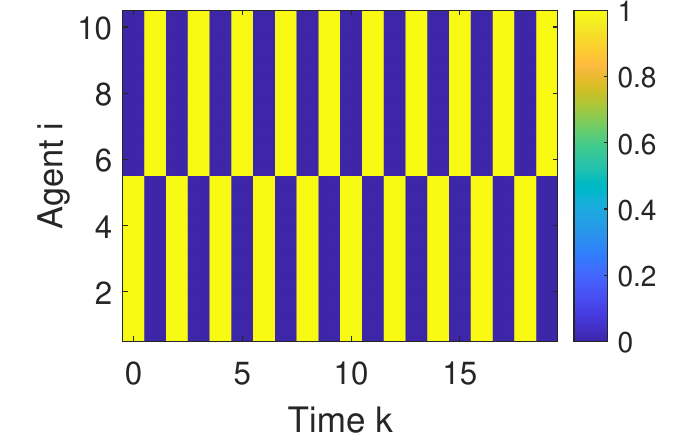}
    \caption{Periodic schedule}
\end{subfigure}

\caption{Comparison of average communication schedules under MPC,
one-shot, auction-based, and periodic strategies.}
\label{fig:transmission_patterns} \vspace{-4 mm}
\end{figure}

Fig.~\ref{fig:transmission_patterns} shows the average transmission
decisions $\theta_{i,k}$ over all Monte Carlo runs. The MPC scheduler
distributes transmissions across agents and time according to the
evolving prediction errors and gradually reduces communication toward
the end of the finite horizon. The one-shot scheduler follows the
transmission sequence computed at the beginning of the horizon and
therefore cannot react to subsequent realizations. The auction-based
scheduler also reallocates communication resources online based on the
current transmission-benefit scores and consequently exhibits an
adaptive transmission pattern. In contrast, the periodic scheduler
follows a fixed alternating pattern that is independent of the system
state.

To assess computational scalability, we additionally compare the online
decision time of the centralized MPC-MILP and the auction-based
scheduler as the number of agents increases. For this experiment, the
network capacity is scaled according to
$C=\lceil N/2\rceil$ so that approximately the same relative
communication capacity is maintained as the network grows.
Fig.~\ref{fig:scalability} reports the average computation time required
for a single scheduling decision. The solid curves represent the mean
computation time over repeated runs, while the shaded regions indicate
$\pm1$ standard deviation.

\begin{figure}[t]
\centering
\includegraphics[
    width=0.70\linewidth,
    trim={140 275 160 285},
    clip
]{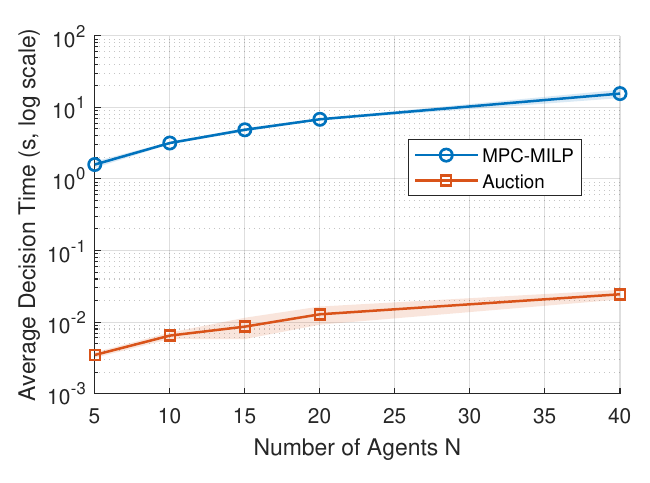}
\caption{Online scheduling time versus number of agents; shaded regions show $\pm1$ standard deviation.}
\label{fig:scalability} \vspace{-4 mm}
\end{figure}

As shown in Fig.~\ref{fig:scalability}, the centralized MPC-MILP
requires computation times on the order of seconds and generally becomes
more expensive as the number of agents grows. In contrast, the
auction-based scheduler remains in the millisecond range, with more than
two orders of magnitude lower computation time across the tested cases.
Together with its total cost being only $2.39\%$ above MPC for
$N=10$ and $C=5$, this demonstrates a favorable
performance-complexity trade-off.

\section{Conclusion}
\label{sec:conclusion}

This paper developed a centralized receding-horizon scheduling framework
for multi-agent LQG control under shared network-capacity constraints,
together with a low-complexity auction-based alternative. Numerical
results show that the MPC scheduler achieves the lowest average total
cost among the considered methods, while the auction-based scheduler
incurs only a $2.39\%$ increase in total cost relative to MPC. The
scalability study further shows that the auction-based scheduler requires
substantially less online computation, remaining in the millisecond range
as the number of agents increases, whereas the centralized MILP requires
computation times on the order of seconds. These results illustrate the
performance-complexity trade-off between centralized optimization and
low-complexity adaptive scheduling. Future work will focus on deriving
analytical suboptimality bounds for the auction-inspired scheduler.

\bibliographystyle{IEEEtran}
\bibliography{main}
\end{document}